\documentclass[sigconf]{aamas}  

\AtBeginDocument{%
  \providecommand\BibTeX{{%
    \normalfont B\kern-0.5em{\scshape i\kern-0.25em b}\kern-0.8em\TeX}}}
    
\usepackage{enumerate}
\usepackage{mathrsfs}
\usepackage{framed}
\usepackage{subfigure}
\usepackage{booktabs}    
\usepackage{flushend} 

\setcopyright{ifaamas}  
\copyrightyear{2020} 
\acmYear{2020} 
\acmDOI{} 
\acmPrice{} 
\acmISBN{} 
\acmConference[AAMAS'20]{Proc.\@ of the 19th International Conference on Autonomous Agents and Multiagent Systems (AAMAS 2020)}{May 9--13, 2020}{Auckland, New Zealand}{B.~An, N.~Yorke-Smith, A.~El~Fallah~Seghrouchni, G.~Sukthankar (eds.)}  

\begin{document}

\title{Collaborative Data Acquisition}  


\author{Wen Zhang}
\affiliation{%
  \institution{ShanghaiTech University}
  \city{Shanghai} 
  \country{China}
}
\email{zhangwen@shanghaitech.edu.cn}
\author{Yao Zhang}
\affiliation{%
  \institution{ShanghaiTech University}
  \city{Shanghai} 
  \country{China}
}
\email{zhangyao1@shanghaitech.edu.cn}
\author{Dengji Zhao}
\affiliation{%
	\institution{ShanghaiTech University}
	\city{Shanghai} 
	\country{China}
}
\email{zhaodj@shanghaitech.edu.cn}

\begin{abstract}
  We consider a requester who acquires a set of data (e.g. images) that is not owned by one party. In order to collect as many data as possible, crowdsourcing mechanisms have been widely used to seek help from the crowd. However, existing mechanisms rely on third-party platforms, and the workers from these platforms are not necessarily helpful and redundant data are also not properly handled. To combat this problem, we propose a novel crowdsourcing mechanism based on social networks, where the rewards of the workers are calculated by information entropy and a modified Shapley value. This mechanism incentivizes the workers from the network to not only provide all data they have but also further invite their neighbours to offer more data. Eventually, the mechanism is able to acquire all data from all workers on the network and the requester's cost is no more than the value of the data acquired. The experiments show that our mechanism outperforms traditional crowdsourcing mechanisms.
\end{abstract}

%


\maketitle


\section{Introduction}
\label{section:intro}
Recent years witnessed the rise and development of deep learning~\cite{abelson-et-al:lecun2015deep}. Many laboratories and companies put emphasis on building neural network applications such as DeepMind, Facebook AI Research (FAIR) and Stanford AI Lab (SAIL). In these applications, large-scale datasets are indispensable. Therefore, data acquisition underpins the success of these applications. Traditionally, they may hire voluntaries to collect data such as photos or voices, which is a very time-consuming and labour-intensive process.

Crowdsourcing is a teamwork collaboration mode in which companies use the open call format to attract potential workers to do the task at a lower cost, which was first proposed by Howe~\cite{gottlob:howe2006rise}. Many companies are committed to crowdsourcing services such as Amazon Mechanical Turk and gengo AI. Consequently, more and more research teams turn to these platforms to acquire data. For example, ImageNet~\cite{abelson-et-al:deng2009imagenet} from SAIL is collected via Mechanical Turk.

In traditional crowdsourcing models, the requester has to pay not only the data providers but also the third-party crowdsourcing platforms. However, the data collected in this way may be redundant, but the requester still has to pay for it. Therefore, whether the requester can benefit from the paid crowdsourcing platforms is not clear.

In this paper, we propose a novel crowdsourcing mechanism for data acquisition via social networks. The requester is the owner of the mechanism and she can use it to collect data without any third-party platforms. The mechanism requires the requester to release the task information to her neighbours on the network. Under this mechanism, the participants will be incentivized to provide all their data and invite all their neighbours to do the task. They will gain payoffs not only from their offered data but also from inviting their neighbours. By doing so, the task information can be disseminated through the whole social network without paying the workers in advance.

Different from other crowdsourcing mechanisms, our mechanism only distributes rewards to those who provide non-redundant data and do effective diffusion. That is, the workers will not gain any payoff if they do not contribute to the data acquisition task. Hence it can eliminate redundant and irrelevant data, and avoid unnecessary expenses for the requester. More importantly, our mechanism can incentivize workers participated to invite all their neighbours to join the task, which is not possible under existing mechanisms. 

In the crowdsourcing literature, there are many related mechanisms published. Franklin \textit{et al.} focused on how to use crowdsourcing to process difficult queries~\cite{abelson-et-al:franklin2011crowddb}. Chawla \textit{et al.} proposed an optimal crowdsourcing contest for high-quality submissions~\cite{abelson-et-al:chawla2015optimal}. Zhou \textit{et al.} studied a new method of measurement principle for work quality~\cite{abelson-et-al:zhou2015regularized}. Miller \textit{et al.} devised a scoring system to evaluate the feedback elicited~\cite{DBLP:journals/mansci/MillerRZ05}. Radanovic\textit{et al.} presented a general mechanism to reward the workers according to peer consistency~\cite{DBLP:journals/tist/RadanovicFJ16}. They are all different from our work. They mainly focused on the crowdsourcing model to improve the quality of the work provided by the workers and their settings have not considered the task propagation between workers. In our setting, we also incentivize the workers to propagate the task information to their neighbours to collect more data. Naroditskiy~\textit{et al.}~\cite{naroditskiy2012verification} initiated a formal study of verification in crowdsourcing settings where information is propagated through referrals. However, there is often a single ground truth in their settings which is unknown to the requester. Our setting is not seeking the answer for a ground truth, and we are aiming for collecting rich data. 

There also exists some interesting literature about information diffusion on social networks. Narayanam and Narahari studied the target set selection problem~\cite{abelson-et-al:Narayanam2010}, which involves discovering a small subset of influential workers in a given social network, to maximize the diffusion quality of the workers rather than incentivizing them to diffuse. In terms of incentivizing people to disseminate the task information, Li~\textit{et al.} proposed a single-item auction diffusion mechanism via social networks and Zhao~\textit{et al.} then generalized the mechanism for multiple items~\cite{DBLP:conf/aaai/LiHZZ17,DBLP:conf/atal/ZhaoLXHJ18}. The problem they studied is a non-cooperative game, while in our setting the workers may benefit from others' participation. Emek~\textit{et al.} studied the reward mechanisms in multi-level marketing within social networks~\cite{DBLP:conf/sigecom/EmekKTZ11}. However, they focused on the false-name manipulations and in their setting, each agent occurring in the referral tree has to purchase the product, which is not required in our setting. Another related work is the MIT winning solution under the DARPA Network Challenge~\cite{abelson-et-al:pickard2011time}. However, their solution only works for tree structures. Our mechanism refers to their idea and puts forward a modified payoff policy for workers' diffusion contribution in single-source directed acyclic graphs. More importantly, the reward in the DARPA network challenge is predefined, while in our setting it varies according to the data offered by the workers.  

Our mechanism is also closely related to the strategy diffusion mechanism proposed by Shen \textit{et al.}~\cite{abelson-et-al:shen2018multi}. However, they focused on the problem of false-name attacks and did not consider data redundancy. Also, their mechanism cannot guarantee that the workers will diffuse the task information to all their neighbours. Winter~\cite{winter1989value} proposed a coalition structure value for level structures. Their idea is similar to our method of evaluating the data contribution. Nevertheless, their structure does not take the priority of different coalitions in the same level into consideration, which is essential for the diffusion incentive in our setting.

The contributions of our mechanism advance the state of the art in the following ways:
\begin{itemize}
	\item We model a crowdsourcing mechanism on social networks without relying on third-party platforms. Our mechanism incentivizes the workers to not only offer their data truthfully but also propagate the task information to all their neighbours without paying them in advance. This guarantees that more non-redundant data will be collected.
	\item We give a novel method to evaluate the non-redundancy of the acquired data and distribute rewards to the workers without unnecessary expenses. This is achieved by a modified Shapley value.
	\item The cost of the requester will be no more than the value of the data acquired and the payoffs are adjustable by the requester, which incentivizes the requesters to apply our mechanism in real-world applications.
\end{itemize}

The remainder of the paper is organized as follows. Section \ref{section:model} describes the model of the problem. Section \ref{section:traditional} shows the challenges for directly extending traditional crowdsourcing mechanism on social networks. Section \ref{section:mech} shows the negative result and gives a description of the proposed mechanism. Section~\ref{Entropy} gives a approach to choose the valuation function for the mechanism. Section \ref{section:prop} analyzes the key properties of the mechanism. Finally, we conduct experiments in Section \ref{section:experiments} and discuss future work in Section \ref{section:conc}.

\section{The Model}
\label{section:model}
Consider a data acquisition task $T$ that is executed on a social network. To simplify the representation, we first model the network as a directed acyclic graph (DAG) $G=(V,E)$ with a single source $s\in V$ which is a special node called the requester of task $T$, and later on we will consider a general graph. In the graph, $V=\{s\}\cup N$ where $N=\{1,\dots, n\}$ denotes the set of $n$ workers and $E$ denotes the information flow between vertices. For any $i\not=j \in V$, if there is a directed edge $e_{ij}\in E$ from $i$ to $j$, then $i$ can directly propagate the task information to $j$. Here, we say $j$ is $i$'s child and $i$ is $j$'s parent. Let $r_i^{c}$ be the set of $i$'s children, $r_i^{p}$ be the set of $i$'s parents and $r_i=(r_i^{c},r_i^{p})$ be the neighbour set of each $i\in V$. If there is a directed path from $i$ to $j$, then we say $j$ is $i$'s successor and $i$ is $j$'s predecessor. For each $i\in V$, let $succ(i)$ be the set of $i$'s all successors, and $pred(i)$ be the set of $i$'s all predecessors. Each worker $i\in V$ has a depth $l_i\geq0$ representing the length of the shortest path from the requester $s$ to $i$. 

In the above network, requester $s$ wants to collect data of task $T$. Each worker $i\in N$ is a potential data owner and has a private dataset $D_i = \{ d_i^1, d_i^2, \dots, d_i^k \}$ related to task $T$, where each $d_i^j \in D_i$ represents an atomic data (e.g. an image) and $k$ is the number of atomic data owned by the worker $i$. Let $\mathcal{D}$ be the space of all possible datasets owned by workers. In our setting, we are not aiming for a single ground truth, instead, we try to collect a dataset as rich as possible.

Given the problem setting, without using crowdsourcing platforms, it is evident that the requester can only collect data among her neighbours with whom she can directly communicate. Traditionally, to collect as many required data as possible, the requester tends to do propagation with the help of some paid third-party crowdsourcing platforms (such as Amazon Mechanical Turk and gengo AI). However, the quality of the data collected cannot be guaranteed and users may tend to give redundant data which is costly but not useful for the requester.

In this paper, we propose a novel diffusion mechanism for crowdsourcing the data. The goal of the mechanism is to incentivize the workers on the social network to provide all the data they have and also propagate the task information to all their neighbours. Different from other data collection platforms, our mechanism does not reward the redundant data providers (i.e., duplicate data will not be paid). Furthermore, the workers' total payoff is relevant not only to their provided data but also to their diffusion contribution (inviting neighbours).

For each worker $i\in N$, let $\theta_i=(D_i,r_i^c)$ be $i$'s type. Due to the information flow constraint, we do not need to consider $r_i^p$ in $i$'s strategy space. Then the type profile of all the workers is denoted as $\theta=(\theta_1, \theta_2,\dots, \theta_n)=(\theta_i, \theta_{-i})$, where $\theta_{-i}$ represents the type profile of all workers except $i$. Let $\Theta_i$ be $i$'s type space, and $\Theta=(\Theta_1,\dots, \Theta_n)=(\Theta_i, \Theta_{-i})$ is the type profile space for all the workers. 

Our mechanism requires each worker $i\in N$ participating in the mechanism to report their type. Worker $i$ may not report her type $\theta_i$ truthfully if it is her interest to do so. Let $\theta_i'=(D_i',{r_i^c}')$ be the type worker $i$ reported, where $D_i'$ is the data $i$ provided and ${r_i^c}'$ is the children $i$ has invited to do the task. Let $\theta_i'=nil$ if worker $i$ is not invited or refuses to participate in the mechanism. In the rest of the paper, we use $\theta'$ to denote the type reports of all workers, which can be different from their true type profile $\theta$.

\begin{definition}	
	Given a report profile $\theta'$ of all workers, let the network generated from $\theta'$ be $G(\theta')=(V',E')\subseteq G$, where $V'=\{s\}\cup\bigcup_{i\in N}{r_i^c}'$ and $E'\subseteq E$ is reduced by $V'$.
\end{definition}

\begin{definition}
	A report profile $\theta'$ is feasible if for each worker $i\in N$ with $\theta_i'\not=nil$, there exists at least one path from requester $s$ to $i$ on the network $G(\theta')$. Given workers' true type profile $\theta$, let $\mathcal{F}(\theta)$ be the set of all feasible report profiles under $\theta$.
\end{definition}

\begin{figure}[htbp]%
	\centering
	\subfigure[A social network.]{%
		\label{net}%
		\includegraphics[width=0.5\linewidth]{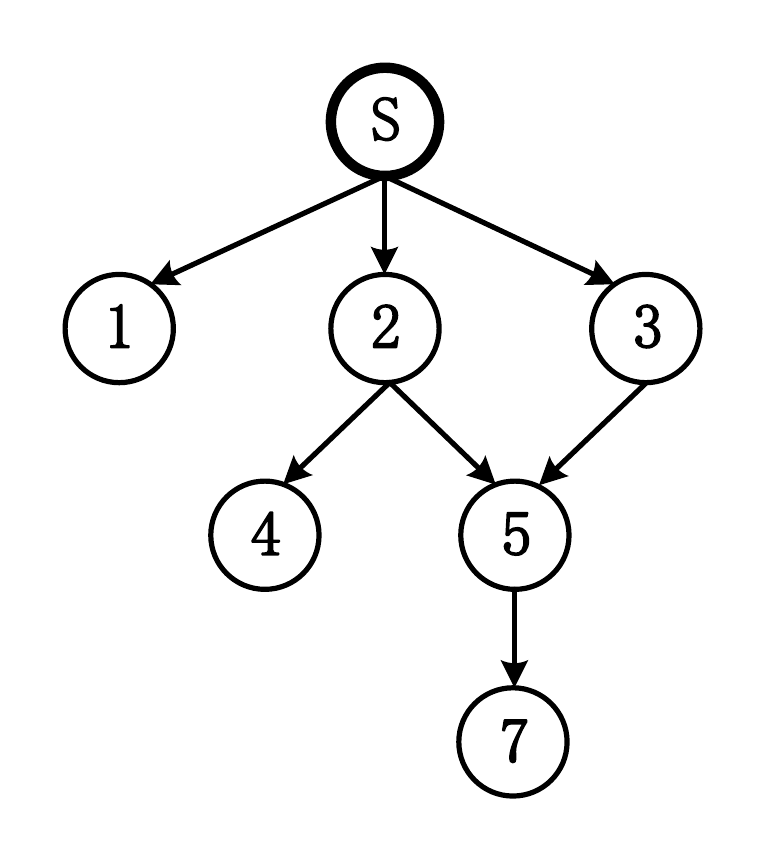}}%
	\subfigure[The generated network.]{%
		\label{gra}%
		\includegraphics[width=0.5\linewidth]{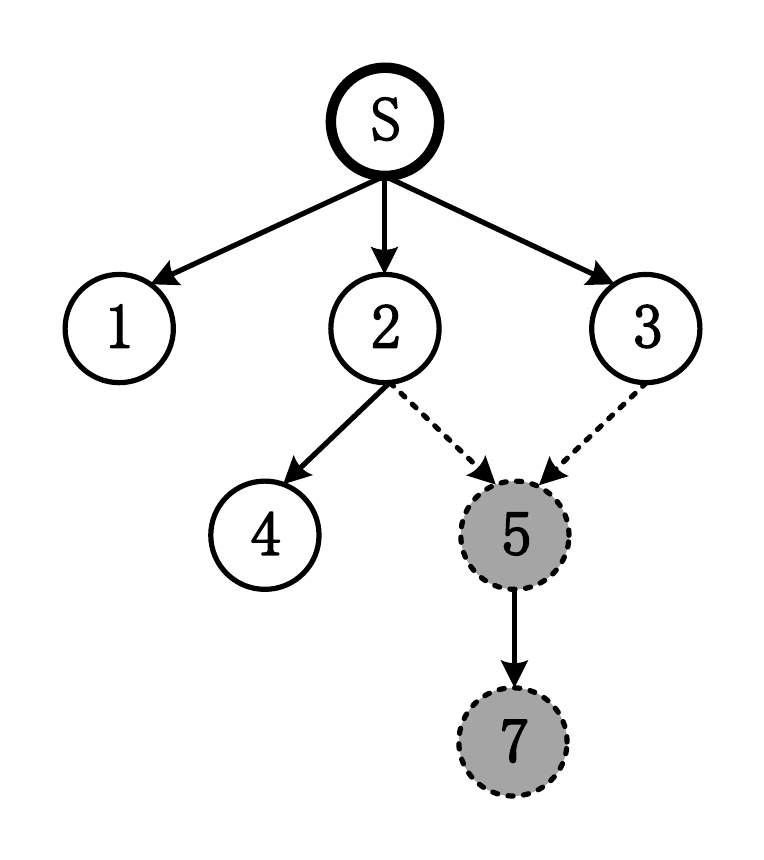}}\\
	\caption{Given ${r_2^c}'=\{4\}$ and ${r_3^c}'=\varnothing$, $\theta_5'$ and $\theta_7'$ must all be $nil$ in any feasible report profile $\theta'\in\mathcal{F}(\theta)$ since worker $2$ and $3$ do not invite $5$.}
	\label{feasible}
\end{figure}

Figure~\ref{feasible} shows an example of feasible report profiles. Intuitively, feasibility means that an agent cannot join in the mechanism if she is not invited/informed about the task, which holds naturally in practice. In other words, infeasible cases will not happen in our mechanism since a worker cannot know the task information if nobody else tells her. Therefore, the following discussion will only focus on feasible report profiles.

In the rest, we define our crowdsourcing diffusion mechanism and its desirable properties.

\begin{definition}
	A crowdsourcing diffusion mechanism  $\mathcal{M}$ on the social network is defined by a payoff policy $p=(p_i)_{i\in N}$, where $p_i:\Theta\mapsto\textbf{R}$. Given a feasible report profile $\theta'\in \mathcal{F}(\theta)$, $p_i(\theta')$ is the payoff of worker $i$ for her data contribution and diffusion contribution. 
\end{definition}

To design a crowdsourcing diffusion mechanism, we hope that workers are incentivized to give all their data and invite all their neighbours to offer more data. This property is called incentive compatibility. An incentive compatible (truthful) diffusion mechanism guarantees that for all workers $i\in N$, reporting her true type is a dominant strategy, i.e., $\theta_i'=(D_i, r_i^c)=\theta_i$.

\begin{definition}		
	A crowdsourcing diffusion mechanism $\mathcal{M}=p$ is \textbf{incentive compatible} (IC) if $p_i(\theta_i,\theta_{-i}')\geq p_i(\theta_i'',\theta_{-i}'')$, for all $i\in N$, all $\theta' \in \mathcal{F}(\theta)$, all $\theta_i''\in\Theta_i$, where for any $j\neq i$, $\theta_j'' = \theta_j'$ if there exists a path from $s$ to $j$ in $G(\theta_i'',\theta_{-i}')$, otherwise $\theta_j'' = nil$.	
\end{definition}

Note that in the IC definition, we need to adjust the reports of $\theta_{-i}'$ when $i$'s report changes because some workers may not know the task information consequently. 

Under the crowdsourcing diffusion mechanism $\mathcal{M}=p$, requester's payment $\mathcal{P}_s$ is the sum of the payments made to the workers. Traditionally, budget constraint requires the requester's payment to be always bounded by a constant. However, in our setting, owing to the objective to acquire as many data as possible, we extend the definition and say $\mathcal{M}$ is budget constrained if $\mathcal{P}_s$ is bounded by the total value of the dataset collected, which is reasonable since the expenditure of the requester will be no more than the value of the data acquired.

\begin{definition}
	A crowdsourcing diffusion mechanism $\mathcal{M}=p$ is \textbf{budget constrained} (BC) if for all $\theta\in \Theta$ and all $\theta' \in \mathcal{F}(\theta)$, we have
	\begin{center}
		$\mathcal{P}_s(\theta')=\sum_{i\in N}p_i(\theta')\leq v(D)$
	\end{center}
	where $v(D)$ is the value of the total dataset $D$ acquired by the requester.
\end{definition}

We say a mechanism is unbounded reward constrained if there is no limitation for a worker's payoff even if the data she owned and the number of her neighbours are limited. To meet the requirement, the mechanism should reward workers for their inviting, which is essential in practice for incentivizing diffusion. 
\begin{definition}
	A crowdsourcing diffusion mechanism $\mathcal{M}=p$ is \textbf{unbounded reward constrained} (URC) if there exists some positive integer $d$ such that for every real $a$, there exists a worker $i$ of maximum number of neighbours $d$ and a feasible reported type $\theta'\in\mathcal{F}(\theta)$ in some social network such that
	\begin{center}
		$p_i(\theta')\geq a$
	\end{center}
\end{definition}

In a data acquisition problem, whether a mechanism can differentiate the redundancy of data is important. A data-redundancy differentiable mechanism will not reward more to those repeated data, which reduces the requester's unnecessary expenditure. Thus we also take it into consideration.

\section{Traditional Crowdsourcing Mechanism}
\label{section:traditional}
Consider the data acquisition problem based on social networks, seemingly the traditional crowdsourcing mechanism can be easily extended to the new setting. However, in this section, we first extend the traditional crowdsourcing mechanism on social networks and then show that the extended mechanism may distribute rewards for redundant data and violate the properties.

A classic crowdsourcing mechanism gives a fixed reward to each worker participating in the task without considering the quality of the data they provide. In this way, the mechanism cannot differentiate agents based on their capabilities and contributions. That is, no matter what data and how many data a worker provides, she will receive a fixed reward which is predefined by the requester. Besides, workers will not be incentivized to give all the data they have since their reward will not increase with the amount of the data they provide.

A simple modification of the above mechanism is to distribute reward according to their work. For example, a fixed reward is predefined for an atomic of data. Then the more data a worker provides, the more reward will be given to her. However, since the budget is constrained and some reward will be given to redundant data, it will not collect enough data for the requester. Moreover, the workers have no incentive to invite their neighbours to do the task as there is no reward for doing so, which violates the unbounded reward constraint property.

Therefore, diffusion contribution should also be rewarded in order to incentivize the workers to inform their neighbours about the task. A trivial method is to set a fixed amount of money to be the bonus pool for inviting their neighbours. Then the money will be shared among all the workers with diffusion contribution by scaling down to meet the constraint of budget. However, it may violate the property of incentive compatibility and unbounded reward constraint since the workers' reward for their diffusion contribution is related to the number of workers who share the bonus pool. Hence, the workers may refuse to invite others in order to share more money. 

The above discussion raises a few questions: How can the mechanism avoid distributing the rewards to those redundant data? How can the mechanism incentivize the workers to diffuse the task information, without sacrificing the property of incentive compatibility, budget constraint and unbounded reward constraint? In the next section, we will introduce our mechanism which can handle all these problems.

\section{Crowdsourcing Diffusion Mechanism} 
\label{section:mech}
In this section, we first show the negative result of mechanism design for data acquisition settings with cost. Then we focus on the cost-free setting and present our novel diffusion mechanism with desirable properties.

\subsection{Impossibility Theorem}

In what follows, we first study the data acquisition setting, where each worker provides her data with some cost. We investigate whether there exists any mechanism that satisfies incentive compatibility, individual rationality (non-negative utility), budget constraint and unbounded reward constraint when the cost for providing data is considered.

Let $c(D_i')$ be the cost of worker $i$ for providing her dataset $D_i'\subseteq D_i$, which is verifiable for the requester. Then, for worker $i\in N$ of type $\theta_i$, given a feasible report profile $\theta'$ of all buyers, $i$'s utility is defined as
\begin{center}
	$u_i(\theta_i,\theta')=p_i(\theta')-c(D_i')$
\end{center}
where $\theta'=(\theta_i',\theta_{-i}')$ and $\theta_i'=(D_i',{r_i^c}')$.

It is natural to require the mechanism to guarantee the non-negative utility for each worker no matter what dataset she provides and how many neighbours she invites. We say the mechanism is individually rational if it satisfies such property.

\begin{definition}
	A crowdsourcing diffusion mechanism $\mathcal{M}=p$ is \textbf{individually rational} (IR) if $u_i(\theta_i,\theta')\geq0$ for all $i\in N$, all $\theta\in\Theta$ and all $\theta'\in \mathcal{F}(\theta)$.
\end{definition}

Now, we show the negative result regarding the mechanism design problem in the setting with cost.

\begin{proposition}\label{impossibility1}
	In the setting with cost, there exists no mechanism which is individually rational and budget constrained.
\end{proposition}

\begin{proof}
	Consider the social network with a requester and $n$ workers. According to the definition of individual rationality, we have $u_i(\theta_i,\theta')=p_i(\theta')-c(D_i')\geq0$. Sum up the equations for all $i\in N$, we can infer that 
	\begin{equation*}
		\begin{aligned}
			&\sum_{i\in N}\left(p_i(\theta')-c(D_i')\right)\geq0\\
			&\sum_{i\in N}p_i(\theta')\geq\sum_{i\in N}c(D_i')
		\end{aligned}
	\end{equation*}
	According to the definition of budget constraint, we have $\sum_{i\in N}p_i(\theta')\leq v(D)$. Thus, we can infer the necessary condition that
	\begin{equation*}
		\begin{aligned}
			n\min_{i\in N} c(D_i')&\leq\sum_{i\in N}c(D_i')\leq v(D)\\
			n&\leq\frac{v(D)}{\min_{i\in N} c(D_i')}
		\end{aligned}
	\end{equation*}
	However, for every $\frac{v(D)}{\min_{i\in N} c(D_i')}$, in which $v(D)$ does not depend on $n$, there always exists a real $m$ such that $n>\frac{v(D)}{\min_{i\in N} c(D_i')}$ for each $n>m$.
\end{proof}

This proposition shows that even if we do not consider IC and URC, there is a trade-off between IR and BC for the setting with cost. To deal with the problem, an alternative method may be preparing extra money to compensate the cost for each worker. Thus, in the following discussion, we only focus on the cost-free setting.

\subsection{The mechanism}
Next, we will introduce our novel crowdsourcing diffusion mechanism (CDM). Under CDM, redundant data will not be rewarded and the workers' reward will increase with the amount of non-redundant data provided. Moreover, the workers are incentivized to diffuse the task information to as many neighbours as possible to gain more reward for their diffusion contribution. The mechanism is also budget constrained.

The payoff policy of CDM is composed of two parts: data contribution and diffusion contribution. The data contribution indicates how the requester validates workers' provided data, and the diffusion contribution indicates how the requester validates workers' diffusion on the social network. Finally, we will give the total payoff policy by applying both.

\subsubsection{Data Contribution}

Since the data-redundancy of differentiability is taken into consideration, an alternative method to evaluate data contribution is Shapley value, which is a classical method to allocate interest in collaborative games~\cite{roth-alvin:roth1988shapley}. Our data acquisition game is a kind of collaborative game. We define $v:\mathcal{D}\mapsto \mathbb{R}^+$ as the valuation function that evaluates the value of a dataset $D$ for the requester. Here the valuation function $v$ should be monotone increasing and bounded, i.e., for datasets $D_x$ and $D_y$, if $D_x \subseteq D_y$, then $v(D_x) \leq v(D_y) < \infty$. 

Then if we directly apply the Shapley value among all workers on the network, the data contribution for each worker $i$ will be:

\begin{equation}\label{Eq:shapley}
\phi_i = \sum_{S\subseteq N \setminus \{i\}} \frac{|S|!(|N| - |S| - 1)!}{|N|!} \left( v(D_{S\cup \{i\}}') - v(D_S') \right)
\end{equation}

Here $D_S$ is the dataset offered by the workers in set $S$: $D_S = \bigcup_{i\in S} D_i$. Intuitively, the Shapley value calculates the average marginal valuation contribution of each worker without considering the network structure. However, with this simple application, workers may not be willing to share the task information with their neighbours.

\begin{proposition}
	A crowdsourcing diffusion mechanism using Shapley value directly as the evaluation of data contribution is not incentive compatible.
\end{proposition}

\begin{proof}
	Consider the network in Figure~\ref{chart:first}, if $D_1 = D_2 = D$ and workers 1 and 2 truthfully offer their data, i.e., $D_1' = D_1$ and $D_2' = D_2$, we have $\phi_1 = \phi_2 = v(D)/2$ according to Equation (\ref{Eq:shapley}).
	
	However, if the worker 1 choose to not propagate the task information to worker 2, then her data contribution becomes $\phi_1' = v(D) > \phi_1$.
\end{proof}

Intuitively, the reason why Shapley value fails is that it divides the rewards equally among all the workers who provide the same data whatever the network structure. Then, the workers will not be willing to invite their neighbours to the task as the neighbours who have the same data will compete with the worker to reduce her payoff, which againsts what we want to achieve with the mechanism. All the other methods which cannot differentiate the invitation relationship will run into such problem.

To combat the diffusion issue with Shapley value, we design a novel payoff sharing policy called \textbf{layered Shapley value}. Let $L_i$ be the set of all the workers with depth $i$: $L_i = \{j| j\in N\ \text{and}\ l_j = i\}$, and $L_i^*$ be all the workers in the first $i$ layers: $L_i^* = \bigcup_{k=1}^i L_k$. Suppose there are totally $K$ layers on the network, then for each worker $i$, the layered Shapley value is defined as follows:

\begin{align}\label{Eq:layeredShapley}
	\hat{\phi}_i & =  \sum_{S\subseteq L_{l_i} \setminus \{i\}} \frac{|S|!(|L_{l_i}| - |S| - 1)!}{|L_{l_i}|!} \cdot \notag \\
	& \left( v\left( D_{L_{l_i-1}^*\cup S \cup \{i\}}' \right) - v\left( D_{L_{l_i-1}^* \cup S}' \right) \right)
\end{align}

Intuitively speaking, Equation (\ref{Eq:layeredShapley}) calculates the average marginal contribution of the workers in the layer using the standard Shapley value, but assumes that all the workers in the prior layers have already joined the coalition before them. More specifically, for the first layer (i.e., the requester's neighbours), the standard Shapley value is applied to calculate their data contribution among the workers in the first layer only. Then for the workers in the second layer, we also apply the Shapley value to compute their data contribution, under the condition that all the workers in the first layer have already been in the coalition. The calculation of workers in the second layer will not change the Shapley value of those in the first layer. This continues for all the other layers. This ensures that workers close to the requester will have a higher priority to get rewards for their data contributions. More importantly, with the layered Shapley value, we can still ensure the following key properties: 
\begin{enumerate}[(1)]
	\item The sum of all workers' layered Shapley value is equal to the valuation of the whole dataset given by workers, i.e. $\sum_{i \in N} \hat{\phi}_i = v(D_{N}')$.
	\item If $i$ and $j$ are two workers in the same layer $L_l$ who are equivalent in the sense that $v(D_{L_{l-1}^*\cup S\cup \{i\}}') = v(D_{L_{l-1}^*\cup S\cup \{j\}}')$ for all $S \subseteq L_l$ $s.t.$ $i,j \notin S$, then $\hat{\phi}_i = \hat{\phi}_j$.
	\item If there is a worker $i$ who has $v(D_{L_{l_i-1}^*\cup S\cup \{i\}}') = v(D_{L_{l_i-1}^*\cup S}')$ for all $S \subseteq L_{l_i}$, which indicates that she does not provide any extra information, then $\hat{\phi}_i = 0$.
\end{enumerate}

Therefore, we will not reward redundant data which has been provided by others in the prior layers. The reason is that in this way child agents cannot decrease the utility of their parents and then all the workers are incentivized to propagate the task information to their neighbours.

\begin{figure}[t]%
	\centering
	\subfigure[]{%
		\label{shapley}%
		\includegraphics[width=0.5\linewidth]{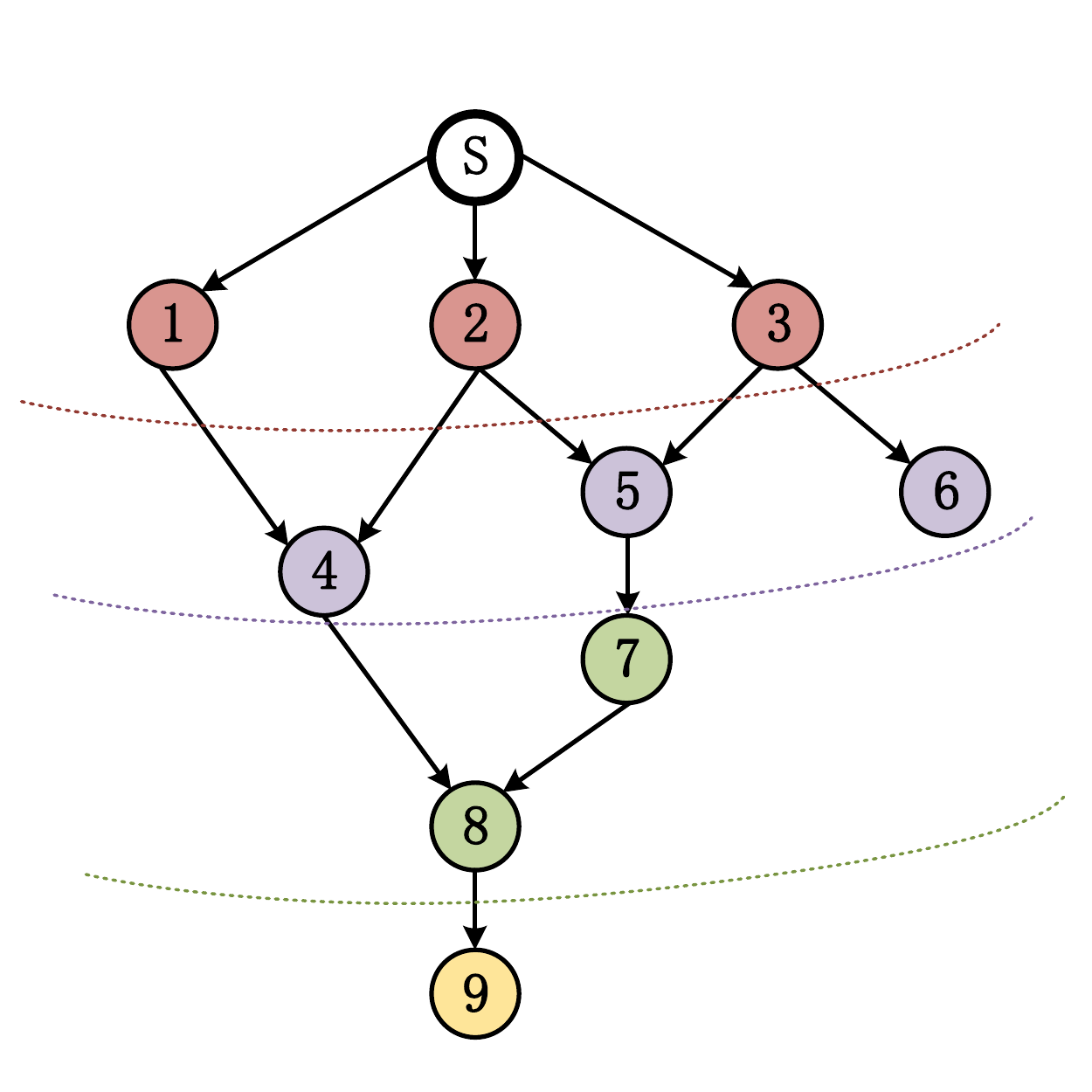}}%
	\subfigure[]{%
		\label{diffusion}%
		\includegraphics[width=0.5\linewidth]{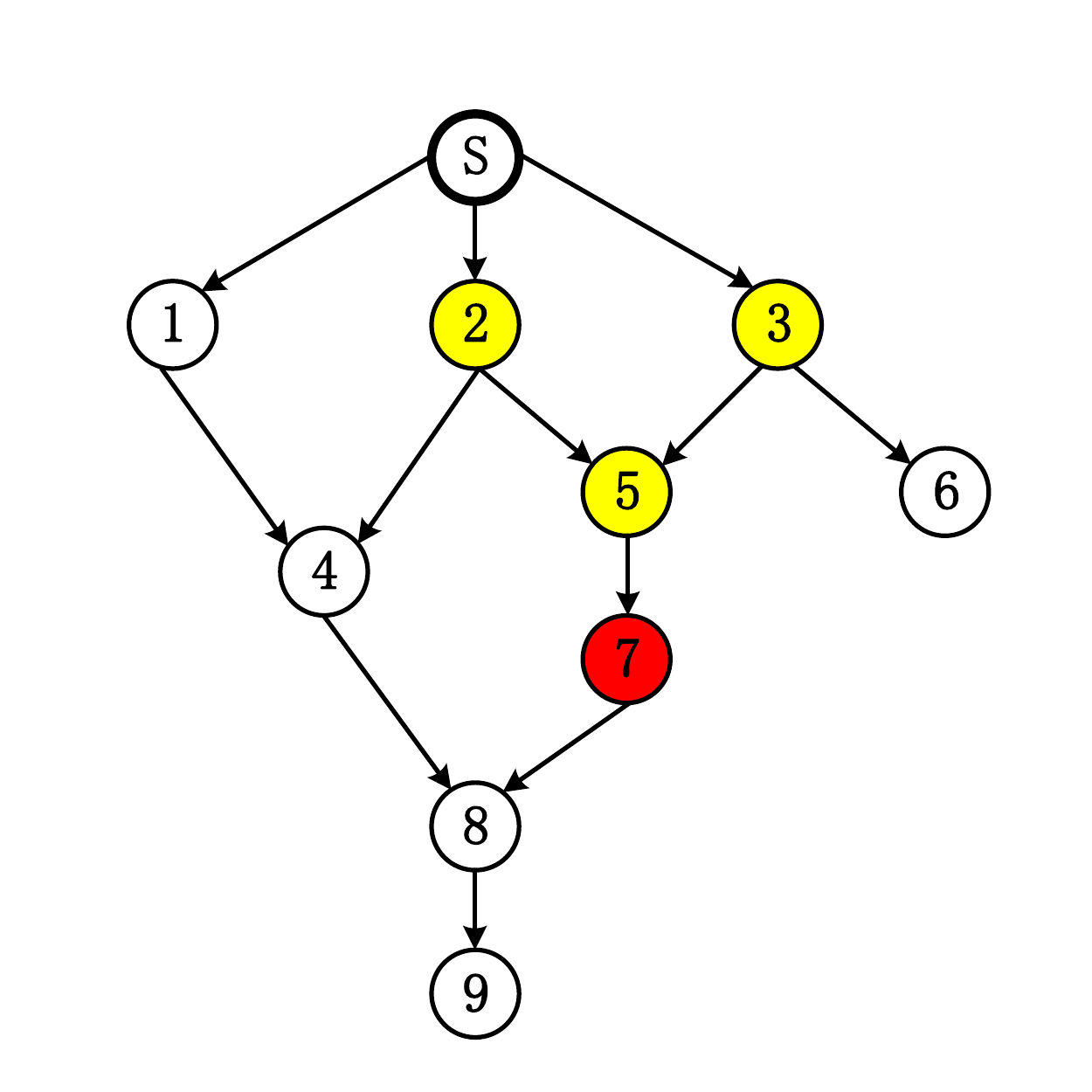}}\\
	\caption{(a) layers in a diffusion network; (b) an example of diffusion contribution}
\end{figure}

Take the network in Figure \ref{shapley} as an example. Worker 1, 2 and 3 are in layer 1; worker 4, 5 and 6 are in layer 2; worker 7 and 8 are in layer 3; worker 9 is in layer 4. The layered Shapley value of worker 1 is: $\hat{\phi}_1 = \frac{1}{6}\cdot ( v(D_1') + v(D_1') + (v(D_{\{1,2\}}') - v(D_2')) + (v(D_{\{1,3\}}') - v(D_3')) + (v(D_{\{1,2,3\}}') - v(D_{\{2,3\}}')) + (v(D_{\{1,2,3\}}') - v(D_{\{2,3\}}'))) = \frac{1}{6} \cdot ( 2v(D_1') -v(D_2') - v(D_3') + v(D_{\{1,2\}}') + v(D_{\{1,3\}}') - 2v(D_{\{2,3\}}') + 2v(D_{\{1,2,3\}}'))$. This is consistent with intuition that what non-redundant data should be.

\subsubsection{Diffusion Contribution}

In traditional crowdsourcing mechanisms, only those who are aware of the task information can compete for some rewards. So the participants who have been informed have no reason to invite their neighbours to do the task. Therefore, to incentivize workers to propagate the information, CDM will give them payoffs for their diffusion contribution. In other words, the workers will gain benefits by spreading the task information to their neighbours effectively.

In our mechanism, the diffusion contribution of a worker $i$ for her successor $j$ is recursively computed as:

\begin{equation}\label{Eq:diff}
\pi_{i,j}=\left\{
\begin{array}{ll}
\sum_{k\in {r_i^p}'}\pi_{k,j}\cdot\gamma\cdot \frac{1}{m_k^j}     &       \text{if	} i\in pred(j)\backslash {s}\\
\alpha \cdot \hat{\phi}_j    &       \text{if } i=s\\
0     &       \text{otherwise}
\end{array} \right.
\end{equation}

where $0<\gamma\leq\frac{1}{2}$ and $0<\alpha\leq1$.

Here, the parameters are interpreted as: $m_k^j$ is the number of worker $k$'s child neighbours which has a path to $j$. For example, in Figure \ref{diffusion}, among all the child neighbours of the requester, only worker $2$ and worker $3$ have a path to worker $7$. Hence, $m_s^7=2$. Similarly, $m_2^7=m_3^7=1$. Factor $\gamma$ is a discount factor and $\alpha$ is the proportion factor, which are predefined coefficients. Note that $\pi_{s,j}$ is a virtual payoff of the requester to simplify the calculation, which will not be paid actually.

\begin{figure}[t]%
	\centering
	\subfigure[]{%
		\label{chart:first}%
		\includegraphics[width=0.245\linewidth]{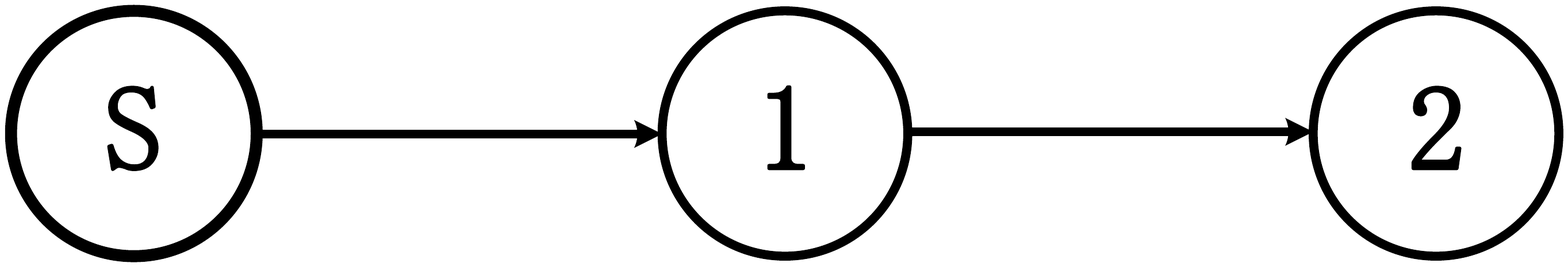}}%
	\subfigure[]{%
		\label{chart:second}%
		\includegraphics[width=0.245\linewidth]{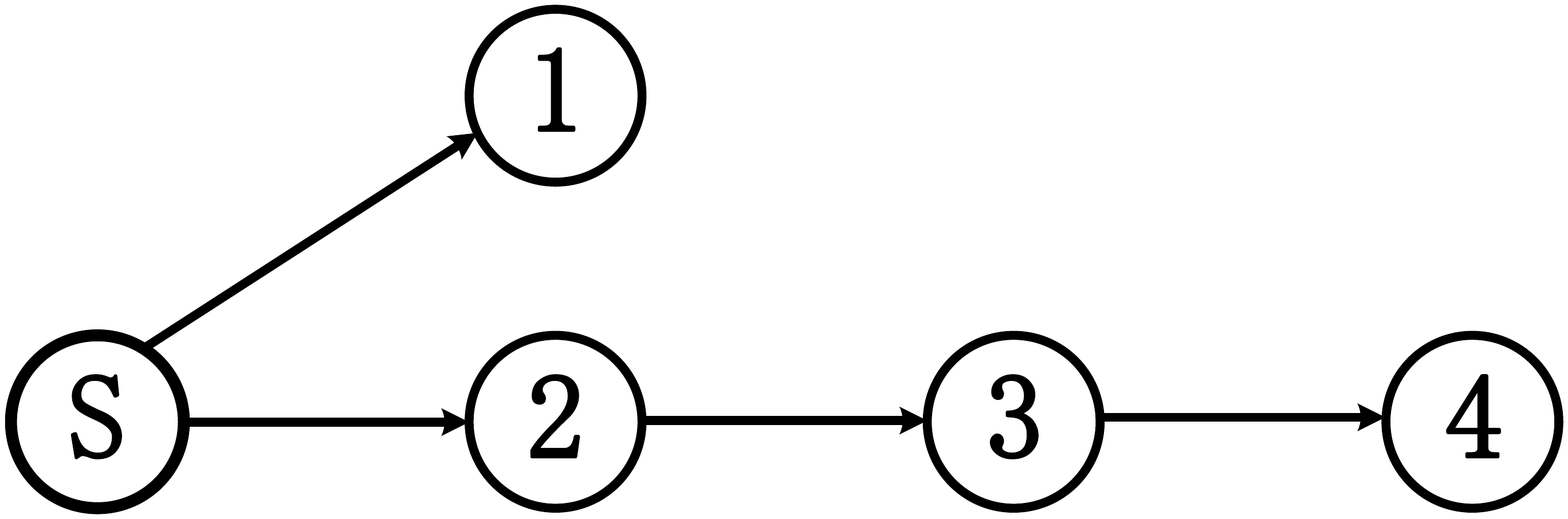}}
	\subfigure[]{%
		\label{chart:third}%
		\includegraphics[width=0.245\linewidth]{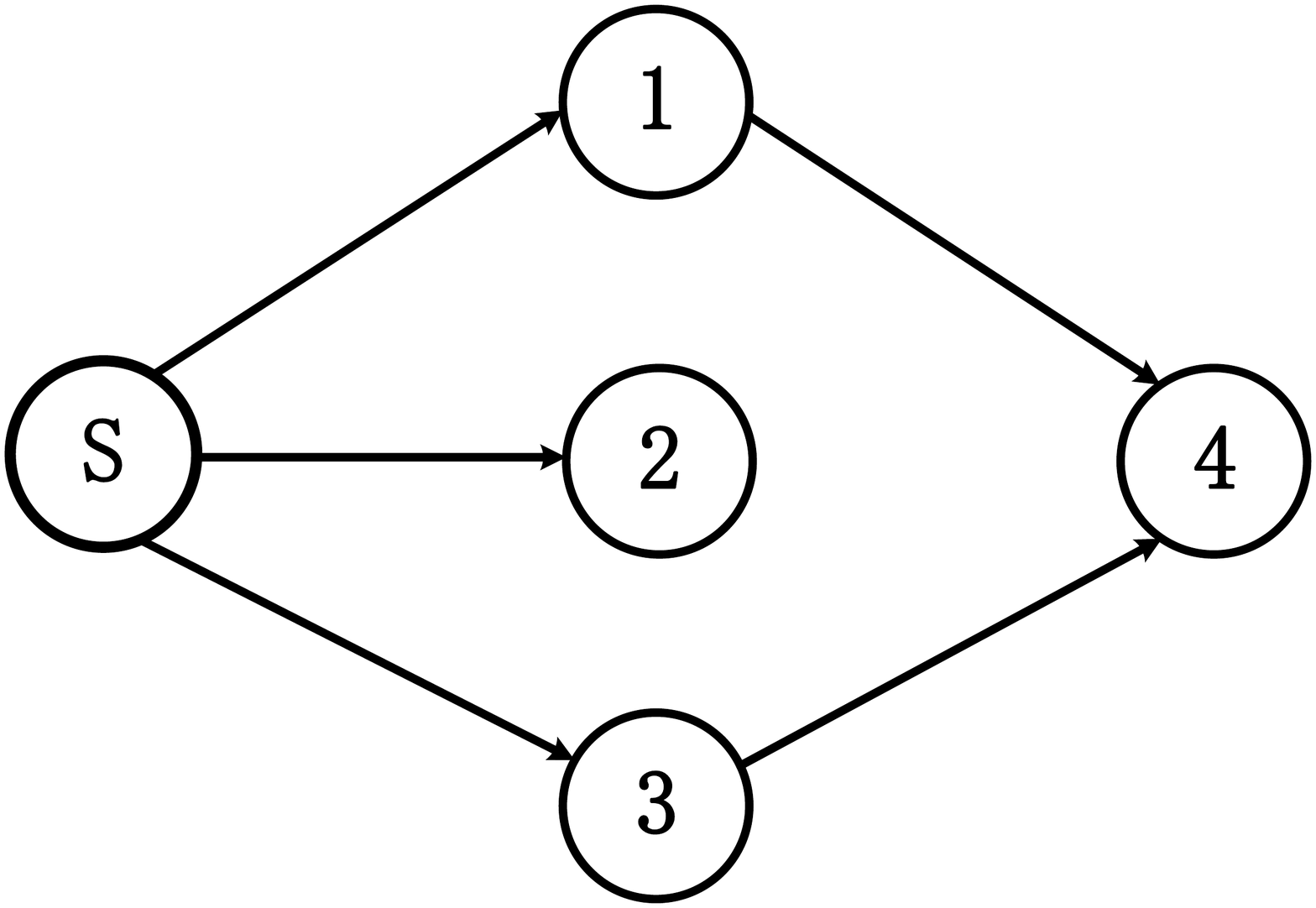}}%
	\subfigure[]{%
		\label{chart:fourth}%
		\includegraphics[width=0.245\linewidth]{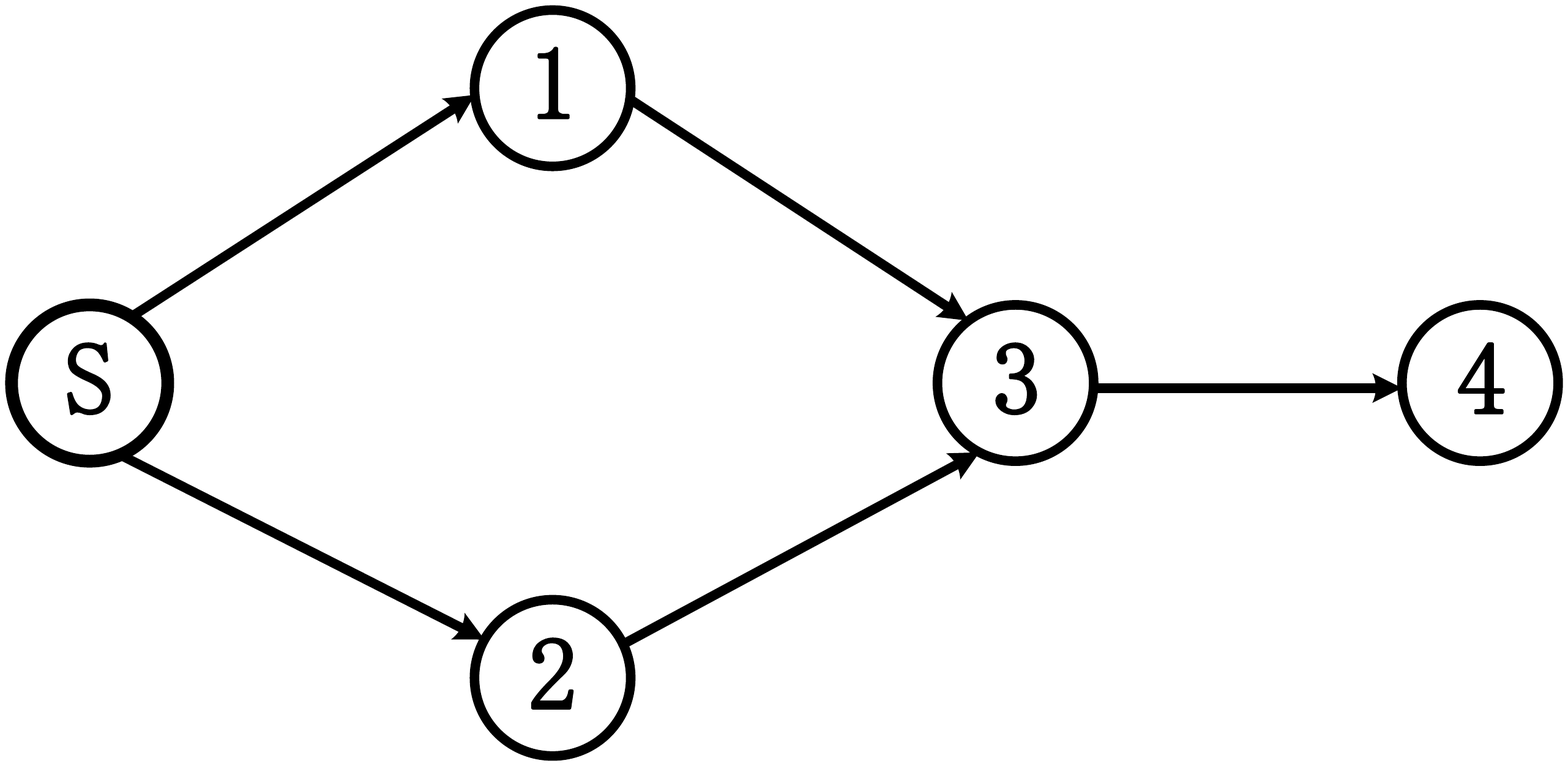}}
	\caption{Basic cases in crowdsourcing diffusion mechanism}
\end{figure}

To show the intuition behind our mechanism, we study three basic cases and only consider the diffusion contribution of worker $3$ for her successor worker $4$ in Figure \ref{chart:second}, \ref{chart:third} and \ref{chart:fourth}. Firstly, we have $\pi_{s,4}=\alpha\cdot\hat{\phi}_4$ for all three cases. In Figure \ref{chart:second}, since the network is a chain, the contribution of a worker is her parent's contribution multiplied by a discount factor $\gamma$, then we have $\pi_{2,4}=\gamma\cdot\pi_{s,4}=\gamma\cdot\alpha\cdot\hat{\phi}_4$ and $\pi_{3,4}=\gamma\cdot\pi_{2,4}=\gamma^2\cdot\alpha\cdot\hat{\phi}_4$. In Figure \ref{chart:third}, since the requester has two children who are connected to worker $4$, the worker $1$ and $2$ have to share the discounted contribution from their parent, then we have $\pi_{1,4}=\pi_{3,4}=\frac{1}{2}\gamma\cdot\pi_{s,4}=\frac{1}{2}\gamma\cdot\alpha\cdot\hat{\phi}_4$. In Figure \ref{chart:fourth}, since the diffusion path from $1$ to $4$ and from $2$ to $4$ both contains worker $3$, worker $3$'s contribution are the sum of the discounted contribution from her parents. Then we have $\pi_{3,4}=\gamma\cdot\pi_{1,4}+\gamma\cdot\pi_{2,4}=\gamma\cdot\alpha\cdot\hat{\phi}_4$. Therefore, all the workers' contribution can be computed by Equation (\ref{Eq:diff}). Finally, the total diffusion contribution of worker $i$ is defined as:
\begin{equation*}
	\pi_i = \sum_{j\in N}\pi_{i,j}
\end{equation*}

The intuition behind the diffusion contribution of CDM is that if a worker's successor provides some non-redundant data, then the worker will be rewarded for her diffusion. Furthermore, from Equation (\ref{Eq:diff}), we can easily conclude that the diffusion contribution is evaluated along the path layer by layer.

The requester can adjust the two factors $\alpha$ and $\gamma$ for different demands. A higher $\alpha$ implies that the requester is willing to give more rewards for diffusion contribution, which will also bring greater expenses. A higher $\gamma$ means that the diffusion contribution will decrease rapidly with depth.

\begin{lemma}\label{lemma:bound}
	Given a data contribution $\hat{\phi}_j$ related to task $T$ from worker $j$, the diffusion contribution distributed to all her predecessors is bounded.
\end{lemma}

\begin{proof}
	According to the definition of diffusion contribution in Equation (\ref{Eq:diff}), we can calculate the total contribution of $j$'s predecessors as: 
	\begin{align*}
		\sum_{i\in N}\pi_{i,j}  \leq \sum_{k=1}^{\infty}\gamma^k\cdot\alpha\cdot\hat{\phi}_j \leq \frac{\gamma}{1-\gamma}\cdot\alpha\cdot\hat{\phi}_j 		
	\end{align*}
	Since $\hat{\phi}_j$ is bounded according to the properties of Equation (\ref{Eq:layeredShapley}), the total contribution of $j$'s predecessors is also bounded.
\end{proof}

Take the network in Figure \ref{diffusion} as an example. Let $\alpha=1$ and $\gamma=\frac{1}{2}$. If worker 7 has a data contribution $\hat{\phi}_7$, then we can calculate all the corresponding diffusion contribution: $\pi_{s,7}=\hat{\phi}_7$; $\pi_{2,7}=\pi_{3,7}=\hat{\phi}_7/4$;
$\pi_{5,7}=\hat{\phi}_7/4$.

\subsection{Total Payoff}
At last, we can get our total payoff policy:
\[ p_i = \lambda \hat{\phi}_i + \mu \pi_i\]
where $0<\alpha\cdot \mu \leq \lambda \leq \frac{1}{2}$ are predefined factors. This is to ensure that the payoff for data contribution is greater than that for diffusion contribution. Otherwise, the workers may not want to offer their data. Another important observation is $p_i\geq 0$ and bounded since $\hat{\phi}_i \geq 0$ and bounded. The detailed proof will be illustrated in Section~\ref{section:prop}.

The total procedure of the mechanism is shown below.

\begin{framed}
	\textbf{Crowdsourcing Diffusion Mechanism (CDM)}
	
	\noindent\rule{\textwidth}{0.35mm}
	
	\noindent\textit{Input}:
	
	A feasible $\theta'\in \mathcal{F}(\theta)$ and parameters $\lambda$, $\mu$, $\alpha$ and $\gamma$ s.t. $0<\alpha\cdot \mu \leq \lambda \leq \frac{1}{2}$, $0<\gamma\leq\frac{1}{2}$ and $0<\alpha\leq1$. 
	
	\begin{enumerate}
		\item Construct the generated social network graph $G(\theta')$.
		\item Run breadth first search on the graph $G(\theta')$ and get the layer sets $L_1$, $L_2$, $\dots$, $L_K$.
		\item For $i$ in $1\dots K$, consider workers in $L_i$:
		
		Compute the layered Shapley value $\hat{\phi}_j$ for each worker $j$ in $L_i$ by Equation (\ref{Eq:layeredShapley}).
		\item Initialize $\pi_{i,j} = 0$ for all $i$, $j \in N$.
		\item For each worker $i\in N$, start from the requester $s$, set $\mathscr{B} = \{s\}$, $\pi_{s,i} = \alpha \hat{\phi}_i$. Until $\mathscr{B} = \{i\}$, do:
		\begin{enumerate}
			\item For each worker $j\in \mathscr{B}$, consider each $k \in {r_j^{c}}' \cap pred(i)$, update the diffusion contribution $\pi_{k,i} \leftarrow \pi_{k,i} + \gamma\cdot \pi_{j,i} /m$, where $m = |{r_j^{c}}' \cap pred(i)|$.
			\item Set $\mathscr{B} = \bigcup {r_j^{c}}' \cap pred(i)$.
		\end{enumerate}
		\item For each worker $i$, calculate $\pi_i = \sum_{j\in N} \pi_{i,j}$.
	\end{enumerate}
	\textit{Output}:
	
	Return total payoff $\lambda \hat{\phi}_i + \mu \pi_i$ for each worker $i$.
\end{framed}

In general, CDM is a centralized data acquisition mechanism. In the beginning, the requester does not know all the workers except her neighbours, so she can only inform her neighbours about the task. Under CDM, the workers informed are incentivized to invite their neighbours to join in the task and to provide all the data they owned to the requester directly. In this way, the requester can know the whole network and collect data as rich as possible without any third-party platforms.

\section{Information Entropy}
\label{Entropy}
Till now, we have qualified the data contribution by the layered Shapley value and presented the mechanism. There is one remaining problem when we apply it to a real-world application, which is how to choose the valuation function $v$. Here we will give a possible approach using information entropy. Information entropy is a function which was first proposed by Shannon~\cite{abelson-et-al:renyi1961measures}. Now it becomes a traditional method to measure the amount of the information of data~\cite{abelson-et-al:Xiao2016Feature,DBLP:journals/entropy/RychtarikovaKMC16}. Information entropy is defined in terms of distributions $\mathbf{q}$ on some space $X$ with finite dimension $|X|$: 

\[ H(\mathbf{q}) = -\mathbb{E}[\log \mathbf{q}] \]

To evaluate a dataset $D$ related to the data acquisition task $T$ by information entropy, we can assume the overall dataset required by the requester $s$ can be classified in $m$ independent target classes, denoted by $\mathcal{X} = \{ \mathcal{X}^1, \dots, \mathcal{X}^m \}$. For each class $\mathcal{X}^j\in \mathcal{X}$, let $X^j$ be its feature space with a predefined finite dimension $|X^j|$. Then for a dataset $D$, every atomic data $d\in D$ can be expressed as a feature vector $d = (x_{d}^1, \dots, x_{d}^m)$, where $x_{d}^j \in X^j$ is the specific feature in class $\mathcal{X}^j$ for $1\leq j\leq m$. For example, if the task is to collect images of nature, let the two target classes be \textit{animals} and \textit{plants}. The space of \textit{animals} is defined as $\{dog,\ cat,\ others\}$ and the space of \textit{plants} is defined as $\{tree,\ flower,\ grass,\ tree\ and\ flower,\ others\}$. Suppose a dataset $D$ has two images $d^1$ and $d^2$, where $d^1$ is an image with a dog beside a tree while $d^2$ is an image with a cat lying on the lawn. Then $d^1 = (dog, tree)$ and $d^2 = (cat, grass)$.

We also need to define a distribution function $Q: \mathcal{D} \mapsto \mathbf{q}$, where $\mathbf{q}=(\mathbf{q}^1,\dots,\mathbf{q}^m)$ is the distribution vector of the dataset $D$. Each $\mathbf{q}^j$ represents the distribution over the feature space $X^j$ of the dataset $D$. In the example above, the distribution of the class \textit{animals} is $\mathbf{q}^1 = (0.5,0.5,0)$ and the distribution of the class \textit{plants} is $\mathbf{q}^2 = (0.5, 0, 0.5, 0, 0)$. Therefore, $Q(D) = ((0.5,0.5,0), (0.5, 0, 0.5, 0, 0))$. 

Now we can use information entropy to evaluate a dataset $D$ using the joint entropy defined on $m$ independent target classes:

\begin{equation}\label{Eq:entropy}
v(D) \triangleq H(Q(D)) = H(\mathbf{q^1},\dots,\mathbf{q^m}) = \sum_{i=1}^{m} H(\mathbf{q^i})
\end{equation}

\begin{lemma}
	Given a dataset $D$ related to task $T$, the valuation of the dataset $D$ by information entropy is bounded.
\end{lemma}

\begin{proof}
	According to the definition of information entropy, we can calculate the valuation of $D$ as: $$v(D) = \sum_{i=1}^m H(\mathbf{q^i}) \leq \sum_{i=1}^m \log |X^i|$$
	
	Since the dimensions of feature spaces of the task $T$ are predefined and finite, the valuation $v(D)$ is bounded.
\end{proof}

\section{Properties of CDM}
\label{section:prop}
In this section, we will prove that our crowdsourcing diffusion mechanism is incentive compatible, unbounded reward constrained and budget constrained. The mechanism also helps the requester collect more non-redundant data. With these properties, a requester is incentivized to apply our mechanism.

\begin{theorem}\label{datasize}
	The data collected from the crowdsourcing diffusion mechanism is no less than only doing the crowdsourcing among the requester's neighbours.
\end{theorem}
\begin{proof}
	Traditionally, the participants in crowdsourcing mechanism are those whom the requester can directly communicate with (i.e., the requester is a platform and participants are the registered users of the platform). These users can be viewed as the requester's child neighbours in CDM, denoted as $r_{s}^{c}\in N$, which is a subset of all the workers on the social network. Then we have: $D_{N}^{CDM}= \bigcup_{i\in N}D_i\supset \bigcup_{r_{s}^{c}}D_i$. Therefore, the amount of data collected in CDM is always equal to or greater than that of traditional crowdsourcing.
\end{proof}	

As is proved in Theorem~\ref{datasize}, more non-redundant data will be acquired by CDM, which incentivizes the requester to apply our mechanism.

\begin{theorem}\label{IC}
	The crowdsourcing diffusion mechanism is incentive compatible.
\end{theorem}
\begin{proof}
	For each worker $i$, her private data $D_i$ is composed of three parts $(D^f_i, D^i_i, D^b_i)$, where $D^f_i$, $D^i_i$ and $D^b_i$ respectively means the data has been offered by the workers in the previous layers, the data can be only offered by the workers in the same layer as $i$ and the data can be offered by the workers in the succedent layers. Obviously, we can discuss the three parts separately.
	\begin{enumerate}
		\item For $D^f_i$, the worker $i$ will receive zero payoffs in our mechanism. She cannot enlarge this payoff by reporting a ${D^f_i}' \subseteq D^f_i$ or by inviting fewer workers since it has nothing to do with the workers in previous layers.
		\item For $D^i_i$, suppose in the layer where $i$ is, there are $k$ workers (including $i$) own this data where $1\leq k \leq |L_{l_i}|$. Then according to the property of Shapley value, if $i$ truthfully offers $D_i^i$, these $k$ workers will share the payoffs for this data. Therefore, the payoff the worker $i$ will receive is $\lambda v(D^i_i)/k$. If she offers a ${D^i_i}' \subseteq D^i_i$, then the payoff will become to $\lambda v({D^i_i}')/k \leq \lambda v(D^i_i)/k$. If she invites fewer workers, it has nothing to do with her payoffs.
		\item For $D^b_i$, suppose worker $i$ is the predecessor of the first worker $j$ in the succedent layers who also owns this data; otherwise, she will not be rewarded if not offering this data or inviting fewer neighbours. If she reports ${D^b_i}' \subset D^b_i$, she transfers some of her data payoffs to diffusion payoffs. Then the payoff for her diffusion contribution is $\mu\pi_{i,j}<\mu\cdot\alpha\cdot\hat{\phi}_j\leq\lambda\hat{\phi}_j=\lambda v(D_i^b-{D_i^b}')$, where $\lambda v(D_i^b-{D_i^b}')$ is the payoff if $i$ offers this part of data by herself. Hence, she will be likely to offer the whole $D_i^b$ by herself.
	\end{enumerate}
	Therefore, for each worker $i$, truthfully reporting her type is the dominant strategy, i.e., $\theta_i'=\theta_i=(D_i,r_i^c)$.
\end{proof}

Theorem~\ref{IC} shows that all the agents' dominant strategy is to provide all the data they owned and invite all their neighbours for their interests. Then we show that workers' reward is unbounded and the requester's expenditure will be no more than the value of the data acquired, which incentivizes both the requester and the workers to take part in the mechanism.

\begin{theorem}
	The crowdsourcing diffusion mechanism is unbounded reward constrained and budget constrained.
\end{theorem}

\begin{proof}
	According to the payoff policy, a worker's total payoff is composed of data contribution and diffusion contribution, which is a monotone increasing function of non-redundant data her descendants provided. Then, a worker's reward will always be increasing as long as her neighbours are invited and they also invite their neighbours. Thus, CDM is unbounded reward constrained.
	
	The total dataset collected by our crowdsourcing diffusion mechanism is $D_{N}=\bigcup_{i\in N}D_i$. In Lemma \ref{lemma:bound}, we have that $\sum_{i \in N}\pi_{i,j}$ is bounded. Since the requester's expenses $\mathcal{P}_s$ is the sum of the payoffs, then we have: 
	
	\begin{equation*}
		\begin{aligned}
			\mathcal{P}_s &= \sum_{i\in N}p_i
			=\sum_{i\in N}(\lambda \hat{\phi}_i + \mu \pi_i)\\
			&=\sum_{i\in N}\lambda\hat{\phi}_i + \sum_{j\in N}\mu\sum_{i\in N}\pi_{i,j}\\
			&\leq\sum_{i\in N}\lambda\hat{\phi}_i + \sum_{j\in N}\mu\cdot\frac{\gamma}{1-\gamma}\cdot\alpha\cdot\hat{\phi}_j \\
			&=\left(\lambda+\mu\cdot\alpha\cdot\frac{\gamma}{1-\gamma}\right)\sum_{i\in N}\hat{\phi}_i\\
			&\leq\left(\lambda+\lambda\cdot\frac{\gamma}{1-\gamma}\right)\sum_{i\in N}\hat{\phi}_i\\
			&\leq\frac{\lambda}{1-\gamma}\sum_{i\in N}\hat{\phi}_i\\
			&\leq2\lambda v(D_{N})\\
			&\leq v(D_{N})
		\end{aligned}
	\end{equation*}
	
	Then, we can conclude that the expenses for a data acquisition task $T$ will not exceed $v(D_{N})$, which is bounded. Moreover, the requester can control the expenses by adjusting the factors.
\end{proof}

At last, we show that our mechanism can work on any social networks rather than DAGs. Since CDM is executed layer by layer, we can first run breadth first traversal on the network and then reduce the edges between the workers in the same layer. After reduction, an arbitrary network can be transferred to a DAG with all the properties remained.

\section{Experiments}
\label{section:experiments}
In this section, we conduct experiments to demonstrate the performance difference between CDM and three classic mechanisms. Our experiments shed light on the advantage of data acquisition through social networks for both the data non-redundancy and the requester's expenditure.

In our experiments, we compare the performance of four mechanisms:
\begin{itemize}
	\item \textbf{NonDiff$\_$eps}: The requester only collects data from her neighbours and distributes each of worker a fixed reward $\epsilon$ as a reward.
	\item \textbf{NonDiff$\_$shapley}: The requester only collects data from her neighbours and calculates each worker's reward by the standard Shapley value.
	\item \textbf{Diff$\_$eps}: The requester collects data from all the workers on social networks and distributes each of worker a fixed reward $\epsilon$.
	\item \textbf{CDM}: The requester collects data from all the workers on social networks and calculates each worker's reward by CDM.
\end{itemize}

We set the number of workers as 15, the size of whole data as 100, the maximum amount of data for each worker as 20, and randomly generate social networks and the data each worker owned. We set $\epsilon$ as the mean of the data from all the workers, which can be viewed as the statistical expectation in real-world applications; we set $\alpha=0.1$, $\gamma=0.5$ and $\lambda=\mu=1$ in the setting of CDM. For each graph, we ran the experiments 20 times. All the experiments were performed using python 3.7 on a machine with a 2.9GHz processor and 16GB RAM.

\begin{figure}[h]
	\centering
	\subfigure[]{%
		\label{Data_Diff}%
		\includegraphics[width=0.5\linewidth]{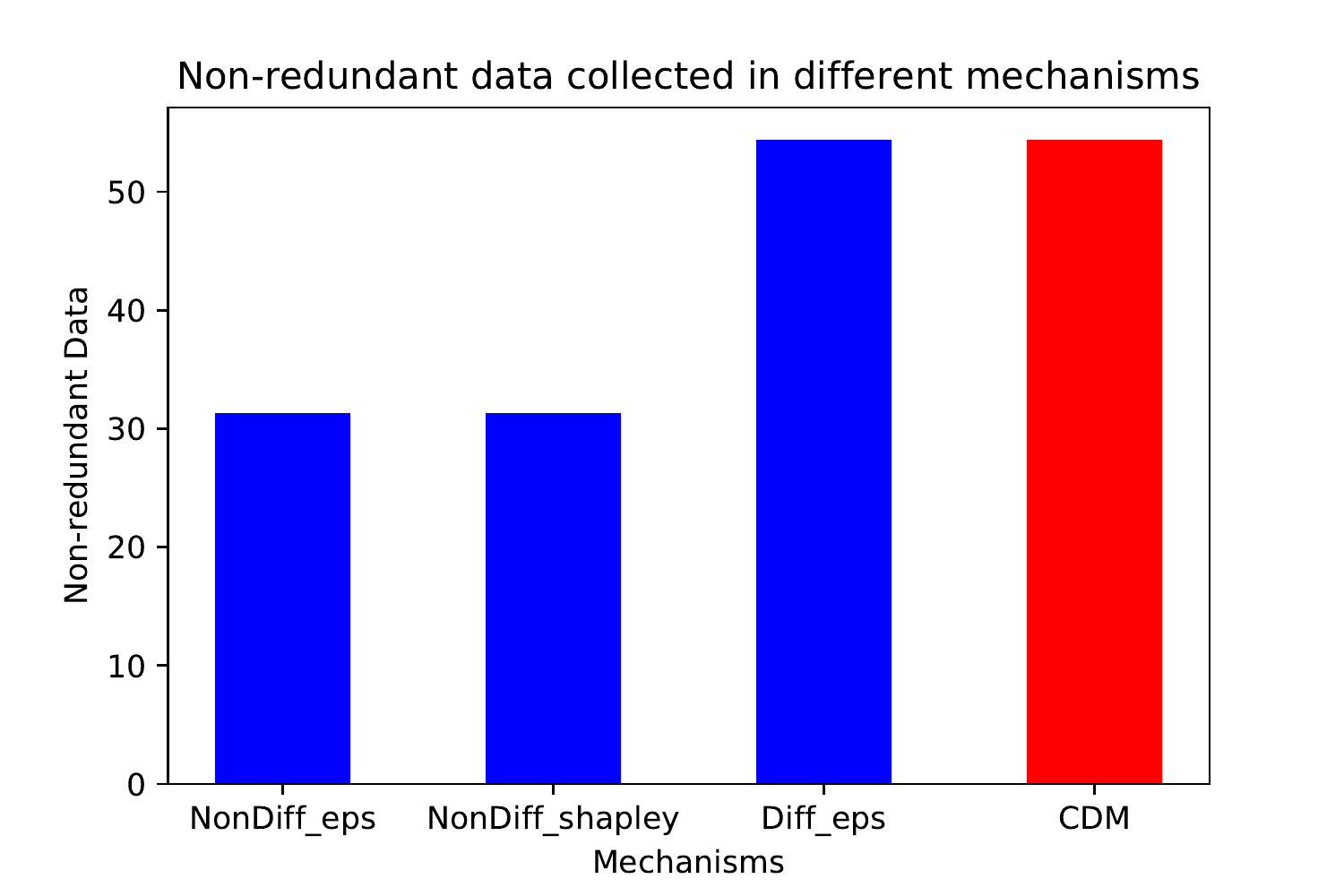}}%
	\subfigure[]{%
		\label{Expenditure}%
		\includegraphics[width=0.5\linewidth]{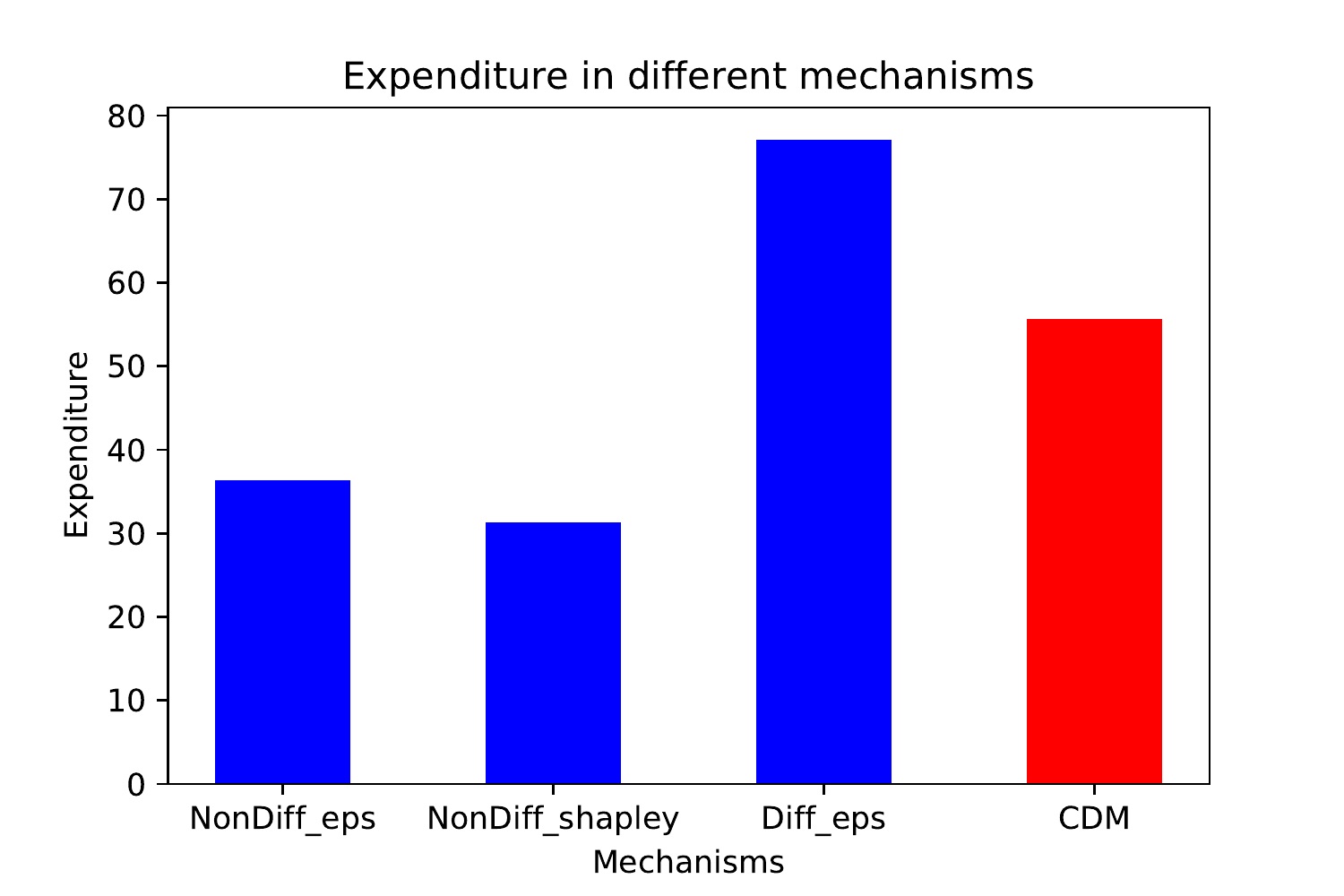}}
	\caption{Given a social network, mechanisms with diffusion collect more data than those without diffusion. For the same amount of data, rewards computed by Shapley value leads to a lower expenditure.}
	\label{fig:experiments}
\end{figure}

In Section~\ref{section:traditional}, we have discussed the limitations for traditional crowdsourcing mechanisms, i.e., violating some theoretical properties. Here we move our attention to the practice of these mechanisms and compare their performance difference. As shown in Figure~\ref{fig:experiments}, experimental results suggest that more workers can be involved in the data acquisition task with diffusion and more data can be collected consequently. Reward distribution with Shapley value can avoid unnecessary expenses to redundant data, which leads to a lower expenditure for the requester, which advances the state of the art for data acquisition tasks.

\section{Conclusion}
\label{section:conc}

In this paper, we have proposed a novel crowdsourcing mechanism via social networks. The mechanism is run by the task requester, and she does not need to pay in advance for the propagation. The prominent contribution of our mechanism is that it incentivizes participants to propagate the task information to their neighbours and to involve more workers in the task. Besides, all workers will also offer as many data as they have. One of the keys to guarantee these properties is that workers close to the requester will have a higher priority to win rewards than their children according to layered Shapley value. We also conducted experiments to further demonstrate the advantages of our mechanism.

Our work has several interesting aspects for future investigation. First of all, the false-name attack is typical in a crowdsourcing system. Hence, designing an advanced mechanism which is false-name proof is a vital successor work. An interesting scene can be considered where workers' action will be affected by their neighbours. Another valuable further work can be generalising our mechanism to other crowdsourcing tasks rather than data acquisition. Although we have shown the impossibility theorem for the setting with cost, it would also be a direction to study the problem after relaxing some assumptions.

%
%


\bibliographystyle{ACM-Reference-Format}  
\bibliography{sample-bibliography}  

\end{document}